\documentclass{article}
\usepackage{arxiv}

\renewcommand{\headeright}{}
\renewcommand{\undertitle}{}

\usepackage{fancyhdr}
\usepackage{array}
\newcolumntype{P}[1]{>{\raggedright\arraybackslash}p{#1}}

\usepackage[utf8]{inputenc}
\usepackage[numbers, sectionbib]{natbib}
\usepackage{amsmath, amssymb, amsfonts, amsthm}
\usepackage{mathtools}
\usepackage{graphicx}
\usepackage{color}
\usepackage[table]{xcolor}
\usepackage{float}
\usepackage{hyperref}
\usepackage{comment}
\usepackage{algorithm, algpseudocode}
\usepackage{enumerate}
\usepackage{enumitem}
\usepackage{subfigure}
\usepackage{times}
\usepackage{sidecap}
\usepackage{tcolorbox}
\usepackage{blindtext}
\usepackage[symbol]{footmisc}

\usepackage{boxedminipage}
\usepackage{xspace}

\newcolumntype{C}[1]{>{\centering\arraybackslash}p{#1}}

\def \polylog{\operatorname{polylog}}

\theoremstyle{plain}
\newtheorem{theorem}{Theorem}[section]
\newtheorem{lemma}[theorem]{Lemma}

\theoremstyle{definition}
\newtheorem{definition}[theorem]{Definition}

\newtheorem{remark}[theorem]{Remark}

\newtheorem*{theorem*}{Theorem}
\newtheorem*{definition*}{Definition}
\newtheorem*{prfthm*}{Proof of Theorem}

\newcommand{\etal}{\textit{et al. }}
\newcommand{\remove}[1]{}

\AtBeginDocument{%
  }

\title{Improved Linear-Time Construction of Minimal Dominating Set via Mobile Agents}

\author{
 Prabhat Kumar Chand \\
  Indian Statistical Institute\\
  Kolkata, India \\
  \texttt{pchand744@gmail.com}
  \And
 Anisur Rahaman Molla\thanks{A.~R.~Molla is supported, in part, by ANRF-SERB Core Research Grant, file no.~CRG/2023/009048, and R.~C.~Bose Centre's internal research grant.} \\
  Indian Statistical Institute\\
  Kolkata, India \\
  \texttt{molla@isical.ac.in}
}

\date{}

\begin{document}

\maketitle

\begin{abstract}
Mobile agents have emerged as a powerful framework for solving fundamental graph problems in distributed settings in recent times. These agents, modelled as autonomous physical or software entities, possess local computation power, finite memory and have the ability to traverse a graph, offering efficient solutions to a range of classical problems. In this work, we focus on the problem of computing a \emph{minimal dominating set} (mDS) in anonymous graphs using mobile agents. Building on the recently proposed optimal dispersion algorithm~\cite{optimal_disc_sync} on the synchronous mobile agent model, we design two new algorithms that achieve a \emph{linear-time} solution for this problem in the synchronous setting. Specifically, given a connected $n$-node graph with $n$ agents initially placed in either rooted or arbitrary configurations, we show that an mDS can be computed in $O(n)$ rounds using only $O(\log n)$ bits of memory per agent, without using any prior knowledge of any global parameters. This improves upon the best-known complexity results in the literature over the same model. In addition, as natural by-products of our methodology, our algorithms also construct a spanning tree and elect a unique leader in $O(n)$ rounds, which are also important results of independent interest in the mobile-agent framework.

\end{abstract}

\section{Introduction}

The use of autonomous agents to solve graph problems has recently attracted significant attention. Such agents, representing entities like self-driving cars, drones, robots, or distributed processes, combine two defining capabilities: they can perform local computations under strict memory constraints, and they can traverse networks, moving between nodes while retaining only limited information. A crucial observation in this model is that local computation cost is essentially negligible compared to movement, as in real-world scenarios where the cost of physical traversal (for example, a self-driven car traversing across mutiple cities) far outweighs local processing. Consequently, research in this area has focused on minimising movement while still enabling efficient solutions to classical graph problems.  

Several fundamental graph problems, such as computing minimal dominating sets and independent sets, leader election, spanning tree construction, and community detection, have been extensively studied both in the classical distributed model and, more recently, in the mobile-agent model. For instance, dominating set construction has been investigated in the mobile-agent setting~\cite{run_for_cover_prabhat} and refined in subsequent works~\cite{disc_mst,near_linear_leader,manish_icdcit}, while the closely related maximal independent set (MIS) problem has also been explored~\cite{mis}. The same framework has produced algorithms for spanning structures, including BFS trees~\cite{agent_bfs,tree_agent_prabhat}, MSTs~\cite{disc_mst,manish_icdcit}, and general spanning trees~\cite{butterfly_spaa_prabhat}. These developments have further led to increasingly efficient approaches for leader election. Notably, in the mobile-agent perspective, these problems add additional significance: spanning trees construction enables effective communication between these autonomous entities and helps in global information dissemination, while dominating sets highlight structurally critical nodes, allowing non-essential agents to halt—thereby reducing deployment and movement costs, a highly desirable property in practical, resource-constrained applications. A possible research possibility in this direction has been highlighted in Section~\ref{conclusion}.

Recently,~\cite{optimal_disc_sync} introduced an \emph{optimal} algorithm for the \emph{dispersion problem}, where agents are repositioned so that each node hosts at most one agent. Building on this, we design a new algorithm to compute a \emph{minimal dominating set} of a graph~$G$. In our setting, with $n$ agents on an $n$-node graph, we show that it is possible to achieve multiple objectives simultaneously: deriving a termination condition post dispersion, constructing a spanning tree, electing a leader, and computing an mDS, all within \emph{linear time}. This work improves upon the best-known complexity bounds in the literature and provides the first unified linear-time approach to these fundamental tasks in the mobile-agent model.





\subsection{Contributions}
The first study of the mDS problem in the mobile-agent framework appeared in~\cite{run_for_cover_prabhat}. Subsequent works~\cite{disc_mst,near_linear_leader} did not directly improve its techniques; progress arose mainly from advances in \emph{leader election}, which in turn enabled faster mDS computation or reduced prior knowledge requirements. In~\cite{manish_icdcit}, the authors proposed a logarithmic-time probing method to improve the search before colouring a node, but their approach required knowledge of $n$ and $\Delta$, as well as leader election before mDS construction.

In this paper, we focus on improving the complexity through a new approach. Our algorithms build on the optimal dispersion procedure of~\cite{optimal_disc_sync}, but the construction of an mDS is neither immediate nor straightforward. We develop a colouring-based mechanism and a careful strategy for releasing colours during execution. For arbitrary initial configurations, we further design a novel technique to detect the completion of dispersion before reducing the problem to the rooted case. In contrast, previous works rely on an infinite (or sufficiently large) waiting time to allow dispersion without explicit termination detection (eg, \cite{run_for_cover_prabhat}) or perform expensive computation to first elect a leader (eg, ~\cite{manish_icdcit}) before the actual mDS construction. Moreover, in the arbitrary dispersion algorithm of~\cite{optimal_disc_sync}, termination detection remains challenging since the number of agents $k$ may be significantly smaller than $n$.  Our main contribution is the following result.

\begin{theorem*}
Let $G$ be a simple, connected, anonymous graph with $n$ nodes. Given $n$ mobile agents initially placed over the nodes either in a rooted or an arbitrary configuration, there exists an algorithm through which the agents can collectively compute a minimal dominating set in $O(n)$ rounds, using $O(\log n)$ bits of memory per agent, where $\Delta$ is the maximum degree of the graph.
\end{theorem*}

As natural by-products, our algorithms also achieve spanning tree construction, gathering and leader election within the same time and memory bounds, problems that are of independent interest in distributed agent-based computing. A comparison with prior results is summarised in Table~\ref{tbl:comparison}.


\begin{table}[!t]
\centering
\footnotesize
\begin{tabular}{|C{3.4cm}|C{1.75cm}|C{3.25cm}|C{3.25cm}|C{2.0cm}|}
\hline
{\bf Algorithm} & {\bf Knowledge} & {\bf Time} & {\bf Memory/Agent} & {\bf Initial Config.}\\
\hline
\rowcolor{gray!20}
\multicolumn{5}{|c|}{\textbf{Minimal Dominating Set}}\\
\hline
\textbf{Section~\ref{rooted}} & $-$ & $O(n)$ & $O(\log n)$ & Rooted\\
\textbf{Section~\ref{arbitrary}} & $-$ & $O(n)$ & $O(\log n)$ & Arbitrary\\
Kshemkalyani {\it et al.}~\cite{manish_icdcit} & $n,\Delta$ & $O(n\log\Delta)$ & $O(\log n)$  & Arbitrary\\
Kshemkalyani {\it et al.}~\cite{near_linear_leader} & $-$ & $O(n\log^2n+m)$ & $O(\log n)$  & Arbitrary \\
Kshemkalyani {\it et al.}~\cite{disc_mst} & $-$ & $O(m)$ & $O(n\log n)$ & Arbitrary \\
Chand {\it et al.}~\cite{run_for_cover_prabhat} & $\lambda,\Delta,n,m,\ell$ & $O(m+\ell\Delta\log\lambda+n\ell)$ & $O(\log n)$  & Arbitrary \\
Chand {\it et al.}~\cite{run_for_cover_prabhat} & $-$ & $O(m)$  & $O(\log n)$ & Rooted \\
\hline
\rowcolor{gray!20}
\multicolumn{5}{|c|}{\textbf{Other Results (From Section~\ref{arbitrary})}}\\
\hline
\textbf{Leader Election} & $-$ & $O(n)$ & $O(\log n)$ & Arbitrary\\
\textbf{Gathering ($n$ agents)} & $-$ & $O(n)$ & $O(\log n)$ & Arbitrary\\
\textbf{Spanning Tree} & $-$ & $O(n)$ & $O(\log n)$ & Arbitrary\\
\hline
\end{tabular}

\caption{Comparison of prior and recent results for the minimal dominating set problem.}
\label{tbl:comparison}
\end{table}

\subsection{Related Work}
The first efficient distributed implementation of the dominating set problem in the CONGEST model was studied by Jia \etal~\cite{JRS02}, who refined the greedy strategy of~\cite{LH00} to design a randomized algorithm running in $O(\log n \log \Delta)$ rounds, producing a $\ln(\Delta)$-approximation with only a constant number of messages exchanged per edge. Sultanik \etal~\cite{SSR10} addressed the art gallery problem—equivalent to finding a minimal dominating set in visibility graphs—via a distributed algorithm that runs in time proportional to the graph’s diameter and guarantees a constant-factor approximation with high probability. Kuhn and Wattenhofer~\cite{KW03} proposed LP-based algorithms that compute a dominating set within a factor $(k\Delta^{2/k}\log\Delta)$ of optimal in $O(k^2)$ rounds, with $O(k^2\Delta)$ messages per node; setting $k$ constant yields the first constant-round, non-trivial approximation. Complementing this, Kuhn \etal~\cite{KTW04} established lower bounds showing that even polylogarithmic approximations for dominating set or vertex cover require at least $\Omega\left(\sqrt{\frac{\log (n)}{\log (\log (n))}}\right)$ and $\Omega\left(\sqrt{\frac{\log (\Delta)}{\log (\log (\Delta))}}\right)$ rounds, respectively. More recently, Jiang \etal~\cite{JKY19} developed deterministic algorithms achieving an approximation factor of $(1+\epsilon)(1+\log(\Delta+1))$ in $O(2^{O(\sqrt{\log n \log \log n})})$ and $O(\Delta\polylog(\Delta)+\polylog(\Delta)\log^\star n)$ rounds for $\epsilon > 1/\polylog(\Delta)$, and extended their methods to connected dominating sets.

In the context of mobile agents, Kaur \etal~\cite{d2d_tanvir} introduced a related problem, called \emph{Distance-2-Dispersion} (\textsc{D-2-D}) problem, where $k$ agents settle on nodes subject to two constraints: no two agents may occupy adjacent nodes, and an agent may reuse a node only if no unoccupied node remains that satisfies the first condition. They showed that with $O(\log \Delta)$ memory per agent, the problem can be solved in $O(m\Delta)$ rounds without prior knowledge of $m,n,$ or $\Delta$, and when $k \geq n$, the settled agents form a maximal independent set. The problem of constructing a minimal dominating set (mDS) with mobile agents was first studied by Chand \etal\cite{run_for_cover_prabhat}, who showed that from a rooted configuration, an mDS can be identified in $O(m)$ rounds, while for arbitrary configurations the construction requires $O(\ell \Delta \log \lambda + n\ell + m)$ rounds, assuming prior knowledge of $m,n,\Delta,\lambda$ and the number of clusters $\ell$. They also obtained an $\ln(\Delta)$-approximate minimum dominating set from dispersed configurations. Subsequent works~\cite{disc_mst,near_linear_leader,manish_icdcit} improved these results by removing global knowledge requirements or optimising the trade-offs between time and memory, although the central focus of these works was on the leader election problem.

\subsection{Our Model}\label{sec:model}

\textbf{Graph: }We have an underlying graph $G(V,E)$ that is connected, undirected, unweighted and anonymous with $|V| = n$ nodes and $|E| = m$ edges. Nodes of $G$ do not have any distinguishing identifiers or labels. These nodes do not possess any memory and hence cannot store any information. The degree of a node $v\in V$ is denoted by $\delta(v)$ and the maximum degree of $G$ is $\Delta$. Edges incident on $v$ are locally labelled using port numbers in the range $[0,\delta(v)-1]$. The edges of the graph serve as \emph{routes} through which the agents can commute. Any number of agents can travel through an edge at any given time. \\ 

\noindent\textbf{Mobile Agents: }A collection of $n$ agents enumerated as $\mathcal{R} = \{r_1,r_2,\dots,r_n\}$ resides on the nodes of the graph with each having a unique ID $\in$ $[0,n^{O(1)}]$. We assume that the highest ID among the $n$ agents is denoted by $\lambda$ with ($\lambda\leq n^{O(1)}$). An agent retains and updates its memory as needed. Two or more agents can be present (\emph{co-located}) at a node or pass through an edge in $G$. However, an agent is not allowed to stay on an edge. An agent can recognise the port number through which it has entered and exited a node. The agents do not have any visibility beyond their (current) location at a node. An agent at a node $v$ can only realise its adjacent ports (connecting to edges) at $v$. Only the collocated agents at a node can sense each other and exchange information. An agent can exchange all the information stored in its memory instantaneously. For colouring, each agent maintains a variable indicating its colour, chosen from $\{\text{\textcolor{red}{red}, \textcolor{blue}{blue}}\}$. \\

\noindent \textbf{Communication Model: }We consider a synchronous system where the agents are synchronised to a common clock and the {\em local communication} model, where only co-located agents (i.e., agents at the same node) can communicate among themselves. In each round, an agent $r_i$ performs the $Communicate-Compute-Move$ $(CCM)$ task-cycle as follows: (i) {\em Communicate:} $r_i$ may communicate with other agents at the same node, (ii) {\em Compute:} Based on the gathered information and subsequent computations, $r_i$ may perform all manner of computations within the bounds of its memory, and (iii) {\em Move:} $r_i$ may move to a neighbouring node using the computed exit port. We measure the complexity in two metrics, namely, time/round and memory. The {\em time complexity} of an algorithm is the number of rounds required to execute the algorithm. The {\em memory complexity} is measured w.r.t. the amount of memory (in bits) required by each agent for computation.

\subsubsection{Problem Statement:}
Let $G(V,E)$ be a simple, connected, anonymous graph with $|V|=n$. Suppose $n$ autonomous agents are initially distributed arbitrarily over the nodes of $G$. The goal is to design an algorithm that repositions these $n$ agents across the nodes of $G$ so as to compute a minimal dominating set of the graph, while minimising both the overall time complexity and the memory required at each agent.



\section{Preliminaries}\label{prelims}

\subsection{A Linear-Time Graph Covering and Dispersion Algorithm}
\label{sec:prelim_optima_disp}
In~\cite{optimal_disc_sync}, the authors solve the dispersion problem in $O(k)$ rounds, where $k\leq n$ agents need to reposition themselves into distinct nodes such that no node hosts more than one agent. Here, we provide a brief description of the algorithm for our model (where $k=n$) for both the rooted and arbitrary configurations. First, we describe the algorithm for the rooted configuration. The algorithm employs a Depth-First Search (DFS) strategy to explore the graph and incrementally settle agents. A key contribution of the work is in addressing the classical bottleneck of DFS-based dispersion: the time spent in searching for an unoccupied neighbour to continue traversal. In earlier approaches, this search was sequential and incurred $O(\Delta)$ time per step, bringing the total edge count $m$ into the time complexity~\cite{KS21}.

Sudo \textit{et al.}~\cite{sudo_near_linear_dispersion} improved this by proposing a parallel probing technique that reduced the neighbour search time to $O(\log \Delta)$ rounds. Their method escalated the search by recursively bringing in agents from settled neighbourhoods in a doubling fashion. Building on this idea, the algorithm in~\cite{optimal_disc_sync} further reduces the probing cost to $O(1)$ rounds. The main insight is to proactively reserve at least $\lceil n/3 \rceil$ agents (called \emph{seeker agents}) for synchronous probing, which allows all neighbours of a node to be probed in parallel. The remaining $\lfloor 2n/3 \rfloor$ agents (called \emph{explorer agents}) are allowed to settle during the DFS.

Maintaining the availability of $\lceil n/3 \rceil$ seekers requires that at least $\lceil n/3 \rceil$ nodes remain unoccupied until DFS finishes. The algorithm ensures this by deliberately leaving certain nodes in the DFS tree empty (we term them as \emph{covered nodes}). A covered node is one that is unoccupied but still accessible—meaning it is "covered" by an agent that can visit it when needed. In contrast, a \emph{fully unsettled node} is both unoccupied and uncovered. The algorithm guarantees that covered nodes are reachable through agent \emph{oscillations}, where a settled agent temporarily moves from its home node to one or more nearby empty nodes and back. In particular, an agent may oscillate between up to three child nodes or between two sibling nodes. The oscillation schedule ensures that if any agent waits at a covered node for six rounds, it is guaranteed to encounter the corresponding oscillating agent. Thus, each settled agent can cover $O(1)$ empty nodes with constant-time oscillation.

Special care is taken around branching points in the DFS tree to decide which nodes should remain empty, which agents should oscillate, and which nodes are permanently settled. This decision-making occurs during the forward and backtrack phases of DFS. The algorithm ensures that every node in the DFS tree is either directly settled or properly covered via oscillation. At each step of the traversal, the agents use the $\lceil n/3 \rceil$ seekers to perform a probing step: if a fully unsettled neighbour is found, the DFS proceeds with a forward move; otherwise, it backtracks. Since there are exactly $n$ forward steps and at most $2(n-1)$ backtracks, the total number of rounds remains $O(n)$.

Once DFS completes and all $n$ nodes have been visited, the reserved $\lceil n/3 \rceil$ seeker agents regroup at the root. They then perform a second traversal of the DFS tree to occupy the previously unfilled nodes. This phase is implemented using a \emph{sibling-pointer} mechanism that allows agents to traverse the tree efficiently with only $O(\log n)$ bits of memory.

In summary, the algorithm achieves $O(n)$ round complexity by separating the responsibilities of settlement and probing: while $\lfloor 2n/3 \rfloor$ agents incrementally settle across the graph, the remaining $\lceil n/3 \rceil$ agents support constant-time probing throughout the traversal. This structural division allows for dispersion with optimal time complexity and significantly improves upon prior techniques that relied on sequential or logarithmic-time searches. Now, for our problem, we modify this algorithm from the point where the last explorer settles and the remaining $\lceil n/3 \rceil$ agents start moving towards the root. At this point, we reach what we call a ``\emph{covered configuration}''. This is the stage where all the $\lfloor 2n/3 \rfloor$ explorer agents have settled. Through this algorithm, we obtain the following state at some point:
\begin{itemize}
    \item \textbf{The team of $\lceil n/3 \rceil$ seeker agents at the root}. These agents can now to used to search all neighbours of a particular node within $O(1)$ rounds. 
    \item \textbf{A \emph{covered configuration}}. In such a configuration, every empty node is either visited periodically (within $6$ rounds) by an oscillating agent or has a permanent settler. 
\end{itemize}To reach this configuration, we simply run the dispersion algorithm from~\cite{optimal_disc_sync} until the point where we have a \emph{covered configuration} and a seeker team gathered at the root.

For the arbitrary starting configuration, dispersion is achieved by combining the tree-subsumption method of Kshemkalyani~\cite{KS21} with the methodology used in the rooted case. Suppose the algorithm begins with $\ell$ clusters. Each of these $\ell$ clusters initiates its own dispersion independently. If the depth-first search (DFS) exploration initiated by a cluster does not encounter any other DFS, it proceeds to completion exactly as in the rooted case. However, if two (or more) DFS processes meet during dispersion, the algorithm ensures that they are merged into the DFS with the larger number of settled agents. Specifically, if DFS $i$ meets DFS $j$, and $i$ currently has more settled agents than $j$, then $j$ is collapsed and all of its agents join the execution of $i$’s DFS. This merging operation is referred to as \emph{subsumption}. Thus, whenever two DFSs meet, the smaller one is subsumed into the larger.  

It is important to note that in the original algorithm, even after dispersion completes, multiple DFS trees may exist, particularly when the number of agents $k$ is significantly smaller than the number of nodes $n$. During DFS construction, agents settle one by one on previously unoccupied nodes, while the remaining agents continue exploring in search of empty nodes. The node currently occupied by all unsettled agents of a DFS, and responsible for further exploration, is referred to as its \emph{head}. A DFS is initially identified by the smallest-ID agent that initiates it, although this identifier may change if the DFS is later subsumed by another. For a given DFS $i$, we denote its head by $head(i)$.  

The dispersion algorithm proceeds in two alternating phases: (i) a \emph{growing phase} and (ii) a \emph{subsumption phase}, which repeat until dispersion is complete. In the growing phase, unsettled agents of a DFS explore new nodes and settle sequentially. In the subsumption phase, if a DFS with $d$ settled agents is subsumed by another, the subsumption process requires $O(d)$ rounds. Consequently, if the initial configuration consists of $\ell$ clusters of sizes $k_1, k_2, k_3, \dots, k_\ell$, the overall time complexity of the dispersion algorithm is  
\[
O(k_1 + k_2 + k_3 + \dots + k_\ell) = O(n),
\]  
where $n$ is the total number of agents. This bound already accounts for the time spent in subsumptions, since the total subsumption cost is $\sum O(d_i) = O(n)$ rounds, where $d_i$ denotes the number of settled agents in DFS $i$ before it is subsumed. Hence, we have the following theorem from~\cite{optimal_disc_sync}.
\begin{theorem}\label{thm:dispersion}
Starting from any initial configuration, dispersion can be solved in $O(n)$ rounds using $O(\log n)$ bits of memory per agent in a synchronous setting.
\end{theorem}

Before presenting our algorithms, we formally define the notion of a minimal dominating set.
\begin{definition}[Minimal Dominating Set]
A subset $D \subseteq V$ of a graph $G = (V,E)$ is called a \emph{dominating set} if every vertex $v \in V \setminus D$ has at least one neighbor in $D$. 
The set $D$ is said to be a \emph{minimal dominating set} (mDS) if no proper subset of $D$ is a dominating set, i.e., removing any vertex from $D$ destroys the domination property.
\end{definition}

\section{Minimal Dominating Set (mDS) from Rooted Configuration}\label{rooted}

In this section, we consider the problem of constructing a minimal dominating set (mDS) in an arbitrary graph using $n$ mobile agents that initially begin at a designated root node. We present an efficient algorithm that completes this task in $O(n)$ rounds, improving upon previous approaches~\cite{run_for_cover_prabhat,manish_icdcit,disc_mst,near_linear_leader}.

\subsection{High-Level Overview} The algorithm proceeds in two main phases. In the first phase, we reposition the $n$ agents into a \emph{covered configuration} (as defined in Section~\ref{sec:prelim_optima_disp}), such that $\lfloor 2n/3 \rfloor$ agents settle across the graph, and the remaining $\lceil n/3 \rceil$ agents---the \emph{seeker agents}---are located at the last visited node, where the final \emph{explorer agent} has just settled, which then traverse back to the root node using the pointers established during the DFS traversal. In the second phase, the agents collaboratively construct a \emph{minimal dominating set} (mDS) of the underlying graph. Our approach builds on the $O(m)$-round algorithm presented in~\cite{run_for_cover_prabhat}. In that work, the root agent is initially assigned the colour \textcolor{red}{red} to indicate its inclusion in the mDS. As the DFS traversal proceeds, each subsequent agent, before settling at a new node, examines the colour of its already-settled neighbours, including its parent. If none of the neighbouring agents are coloured \textcolor{red}{red}, the agent colours itself \textcolor{red}{red}; otherwise, it assigns itself the colour \textcolor{blue}{blue}. This process is inherently sequential, requiring $O(m)$ rounds to complete due to the dependency on local neighbourhood checks at each step.

Our algorithm significantly reduces this complexity by leveraging the $\lceil n/3 \rceil$ seeker agents to parallelise the neighbourhood colour checks. Once the covered configuration is formed, the seeker agents return to the root and re-traverse the graph in a coordinated manner. During this re-traversal, they assist in assigning colours to agents by probing the colours of neighbouring nodes in parallel, enabling each agent to determine its colour in $O(1)$ time. This parallel probing leads to a substantial improvement in overall round complexity.

However, this speed-up introduces a natural question: how do we determine the colour of a node that is currently unoccupied but covered? Since the covered configuration leaves $\lceil n/3 \rceil$ nodes vacant, some nodes may need to be coloured without having a permanently settled agent. We resolve this by using \emph{oscillating agents} - agents that periodically visit such vacant nodes and assign a colour based on the current context of their neighbourhood. This technique avoids the need to settle an agent permanently at such nodes while still ensuring correctness in the colouring process.

\subsection{Details}
The algorithm begins from the \emph{covered configuration}, with the team of $\lceil n/3 \rceil$ seeker agents stationed at the root. The first step is to assign the colour \textcolor{red}{red} to the agent at the root. However, since the agent covering the root may be oscillating, the seeker team might have to wait for a few rounds until the oscillating agent visits the root node. At this point, we recall the two types of oscillation as described in~\cite{optimal_disc_sync}.

\paragraph{Types of Oscillation:}
\begin{itemize}
    \item \textbf{Type A:} Let $u$ be a node with three children $v$, $w$, and $x$ in the DFS tree. Suppose $p(v) < p(w) < p(x)$ are the corresponding port numbers at $u$ connecting to $v$, $w$, and $x$, respectively. In Type A oscillation, the agent at $u$ oscillates in the following sequence: 
    \[
    u \rightarrow v \rightarrow u \rightarrow w \rightarrow u \rightarrow x \rightarrow u \rightarrow \dots
    \]
    In this setup, we say that the agent at $u$ is covering the nodes $v$, $w$, and $x$ through Type A oscillation, and $u$ is referred to as its \emph{home node}.
    
    \item \textbf{Type B:} Let $u$ be a node with parent $\mathsf{parent}(u)$. Suppose $p(\mathsf{parent}(u))$ is the port number at $\mathsf{parent}(u)$ that connects to $u$. Let $v$ and $w$ be two sibling nodes of $u$ such that $v$ and $w$ connect to $\mathsf{parent}(u)$ via port numbers $p(\mathsf{parent}(u)) + 1$ and $p(\mathsf{parent}(u)) + 2$, respectively. Then, a Type B oscillation by the agent at $u$ proceeds as:
    \[
    u \rightarrow \mathsf{parent}(u) \rightarrow v \rightarrow \mathsf{parent}(u) \rightarrow w \rightarrow \mathsf{parent}(u) \rightarrow u \rightarrow \dots
    \]
    In this case, the agent at $u$ covers nodes $v$ and $w$ via Type B oscillation, and again, $u$ is its home node.
\end{itemize}

The type of oscillation an agent follows is determined during the forward and backtrack steps of the dispersion algorithm. Importantly, in both types of oscillation, it suffices to wait at a (possibly empty) node for at most $6$ rounds to ensure that the node is visited (i.e., covered) by some oscillating agent. We now explain how an oscillating agent can simulate or represent the colour of the node it is currently visiting.

\paragraph{Simulating Node Colour via Oscillating Agents:} 
To simulate the colour of multiple nodes visited during oscillation, each agent $r$ maintains a variable tuple:
\[
r.\mathsf{node\_color} = (\mathsf{osc}, \mathsf{color}),
\]
where:
\begin{itemize}
    \item $\mathsf{osc} = 0$ indicates that the agent is at its home node.
    \item $\mathsf{osc} = i\in\{\mathsf{1,2,3}\}$ represents the $i$-th node in its oscillation sequence.
    \item $\mathsf{color} \in \{\textcolor{red}{red}, \textcolor{blue}{blue}\}$ denotes the colour of the node currently being visited.
\end{itemize}

Let us illustrate this using the Type A oscillation pattern described above. Suppose node $u$ is the home node of an oscillating agent $r$, and its children $v$, $w$, and $x$ should have colours: \textcolor{blue}{blue}, \textcolor{red}{red}, and \textcolor{blue}{blue}, respectively. Then, agent $r$ performs the following updates during its oscillation:

\begin{itemize}
    \item At home node $u$: $r.\mathsf{node\_color} \gets (0, \textcolor{blue}{blue})$
    \item Move to $v$: $r.\mathsf{node\_color} \gets (1, \textcolor{blue}{blue})$
    \item Return to $u$: $r.\mathsf{node\_color} \gets (0, \textcolor{blue}{blue})$
    \item Move to $w$: $r.\mathsf{node\_color} \gets (2, \textcolor{red}{red})$
    \item Return to $u$: $r.\mathsf{node\_color} \gets (0, \textcolor{blue}{blue})$
    \item Move to $x$: $r.\mathsf{node\_color} \gets (3, \textcolor{blue}{blue})$
    \item ...
\end{itemize}

For Type-B oscillation, we employ a similar assignment technique, with the $\mathsf{osc}$ value restricted to $\{\mathsf{1,2,3}\}$. Each time an oscillating agent receives its colouring information from the seeker team, it records the colour associated with every node it visits, together with the corresponding $\mathsf{osc}$ value that uniquely identifies the node within its oscillation cycle. To support this process, each agent maintains a constant-sized internal memory capable of storing up to three distinct colour states, which is sufficient for updating the $\mathsf{node\_color}$ variable according to its oscillation pattern. This mechanism guarantees that even vacant nodes—those not permanently occupied—are virtually coloured by the oscillating agents. Having established this, we now proceed to describe our main algorithm.

In the covered configuration, $\lceil n/3\rceil$ seeker agents from the root begin constructing the dominating set. Using the $child$, $sibling$, and $parent$ pointers, they can traverse the graph in $O(n)$ rounds (a technique commonly used; as in~\cite{optimal_disc_sync,butterfly_spaa_prabhat}). When the seeker group meets the home agent at the root, it initiates the colouring process: the root agent sets $\mathsf{node\_color} \gets (0, \textcolor{red}{red})$. The seekers then continue along the DFS, and at the next node, instruct the oscillating agent to set $\mathsf{node\_color} \gets (0, \textcolor{blue}{blue})$.  

At each step, colouring decisions require examining the neighbours. The seekers employ parallel probing, where agents temporarily branch out to visit neighbours, possibly waiting up to $6$ rounds at each node to meet the oscillating agent and collect its colour, if any. After probing, if no neighbour is coloured \textcolor{red}{red}, the current agent is instructed to set its colour to \textcolor{red}{red}; otherwise, it colours itself \textcolor{blue}{blue}. The seekers then resume DFS traversal. The process continues until all nodes have been visited and coloured, after which the seekers return to the root.

In the final phase of the algorithm, the seeker team permanently settles and assigns colours to itself. To accomplish this, the seekers perform a third traversal of the graph from the root. During this traversal, each seeker agent successively occupies one of the remaining vacant nodes and proceeds as follows:

\begin{itemize}
    \item \textbf{If the agent currently at the node is non-oscillating:}  
    In this case, no new agent needs to settle at the node, as it is already being represented by a coloured non-oscillating agent. The seeker team instructs the non-oscillating agent to copy its current colour value into a new permanent variable: $\mathsf{color\_par \gets node\_color.color}$.  
    This variable, $\mathsf{color\_par}$, stores the final (permanent) colour of the agent representing that node.

    \item \textbf{If the agent currently at the node is oscillating:}  
    This situation is further divided into two cases:

    \begin{enumerate}
        \item \textbf{Oscillating agent away from its home node:}  
            One seeker agent (e.g., the one with the smallest ID) permanently settles at the current node and sets $\mathsf{color\_par \gets node\_color.color}$.

        \item \textbf{Oscillating agent at its home node:}  
            The seeker team performs one final oscillation with it, verifies that all covered nodes are permanently occupied and coloured (settling and permanently colouring any remaining nodes as in the previous case using agents from the seeker team), and then returns to the home node. The oscillating agent is finally instructed to permanently settle there with its permanent colour by setting  $\mathsf{color\_par \gets node\_color.color}$.        
    \end{enumerate}The oscillating agent, once it settles, discontinues its oscillation. The newly settled agents correctly update their pointers to maintain the DFS tree structure. 
\end{itemize}

In this way, the seeker team traverses the graph and assigns the final colour to all the $n$ agents representing each node. Since the graph contains $n$ nodes and the total number of agents is also $n$, the $\lceil n/3 \rceil$ seeker agents exactly match the number of vacant positions in the covered configuration. Thus, all remaining nodes are eventually occupied. The correctness and efficiency of this approach follow from the following three lemmas, derived from~\cite{optimal_disc_sync}:

\begin{lemma}\label{lem:lem1optimal}
    In the covering configuration, every non-home empty node is periodically visited (i.e., \emph{covered}) by an oscillating agent from its corresponding home node.
\end{lemma}

\begin{lemma}\label{lem:lem2optimal}
    At the end of achieving the covering configuration, the remaining seeker team contains exactly $\lceil n/3 \rceil$ agents, which matches the number of currently unoccupied nodes in the graph.
\end{lemma}

\begin{lemma}\label{lem:lem3optimal}
    The covering configuration can be achieved in $O(n)$ rounds from a rooted configuration. Additionally, each explorer agent can return to its home node in $O(n)$ rounds, and the seeker team can perform a complete DFS traversal of the graph in $O(n)$ rounds.
\end{lemma}


\begin{theorem}\label{thm:rooted}
    Let $G$ be an arbitrary connected simple anonymous graph with $n$ nodes. Suppose $n$ autonomous mobile agents are initially placed at a designated node (the root) of the graph. Then, the agents can identify a minimal dominating set (coloured \textcolor{red}{red}) in $O(n)$ rounds using only $O(\log n)$ bits of memory per agent.
\end{theorem}

\begin{proof}
    From Lemma~\ref{lem:lem3optimal}, the covering configuration can be constructed in $O(n)$ rounds starting from the rooted configuration. Once this configuration is reached, the $\lceil n/3 \rceil$ seeker agents return to the root and initiate a full DFS traversal of the graph. During this traversal, they assist each explorer agent in determining and fixing its colour based on the colouring rules described earlier. Since each node can be processed in $O(1)$ rounds via parallel probing by the seeker team, and the DFS traversal visits each node only a constant number of times, assigning a colour to all nodes takes $O(n)$ rounds. In the final phase, the seeker agents settle at the remaining vacant positions (i.e., non-home oscillated nodes) and simulate the final colour of those nodes. This final deployment and confirmation of colour values also require $O(n)$ rounds. Therefore, the entire process—from constructing the covering configuration, assigning colours, to completing the minimal dominating set (mDS)—is completed in $O(n)$ rounds.

    The memory required per agent is $O(\log(\Delta + n))$ bits: to store port and neighbour information (max $\Delta$), and to manage $n$ agents. The memory complexity remains consistent with that in~\cite{optimal_disc_sync}, as we use the same variables, with the addition of some extra constant number of variables, $\mathsf{node\_color}$, $\mathsf{color\_par}$, etc., which require only $O(1)$ bits per agent. Hence, the agents correctly and efficiently identify a minimal dominating set in $O(n)$ rounds, using $O(\log (\Delta+n))$ bits of memory per agent; since $\Delta \leq n$, the overall complexity simplifies to $O(\log n)$ bits.
\end{proof}

\section{Minimal Dominating Set (mDS) from Arbitrary Configuration}\label{arbitrary}

We now consider the problem of computing a minimal Dominating Set (mDS) when the $n$ agents are placed in an arbitrary initial configuration, potentially distributed across multiple clusters of the graph. This problem can be reduced to the rooted case (Section~\ref{rooted}) through three stages. First, the agents are dispersed using the algorithm of~\cite{optimal_disc_sync}, which guarantees that with $n$ agents exactly one agent occupies each node. Next, the dispersed agents construct a spanning tree rooted at the node containing the smallest-identifier agent. This is achieved by initially forming several trees, which are then merged until only the least-ID tree remains; the process described in~\cite{butterfly_spaa_prabhat}, which completes in $O(n \log n)$ rounds. Once the tree is established, the agents can gather at the root in time proportional to the tree’s diameter (at most $n$). Finally, from this rooted configuration, an mDS is computed using the method described in Section~\ref{rooted}.

This reduction-based approach, however, faces two main challenges. The dispersion procedure lacks a built-in termination detection, preventing the agents from knowing when to initiate the next stage. Moreover, the spanning-tree construction in~\cite{butterfly_spaa_prabhat} assumes global knowledge of the largest identifier $\lambda$, an extra requirement. To overcome these limitations, we propose an alternative approach that works without any global knowledge and improves the overall complexity of mDS computation from a possible $O(n \log n)$ rounds to $O(n)$ rounds.


\subsection{Details}

Our algorithm proceeds in three key stages:  

\begin{enumerate}
    \item \textbf{Stage 1: Dispersion with Termination Detection.} In this stage, the agents first disperse across the graph. We introduce a novel mechanism that enables all $n$ agents to detect when dispersion has completed. During this process, the agents simultaneously construct a spanning tree of the graph in $O(n)$ time. This result is also of independent interest, as it provides a faster construction than existing spanning-tree algorithms~\cite{butterfly_spaa_prabhat, disc_mst, agent_bfs, tree_agent_prabhat, near_linear_leader} and introduces a new leader election algorithm.  
    
    \item \textbf{Stage 2: Gathering via the Spanning Tree.} Using the spanning tree built in Stage~1, the agents gather at the root node. This step effectively reduces the problem to the rooted configuration considered in Section~\ref{rooted}.  
    
    \item \textbf{Stage 3: Minimal Dominating Set Computation.} Finally, from the rooted configuration, the agents compute a minimal Dominating Set (mDS) following the methodology described in Section~\ref{rooted}.  
\end{enumerate}

\subsubsection{Stage 1: Dispersion with Termination Detection.}


In an arbitrary initial configuration, if the total number of agents distributed across all clusters is significantly smaller than the number of nodes $n$, it becomes challenging for the agents to detect the completion of the dispersion process. The difficulty arises because two different DFS trees, initiated from two separate clusters, may never intersect during dispersion. Since the agents can only communicate locally, no agent within a DFS tree can conclusively determine whether the dispersion process has terminated.  

However, when the number of agents is at least $n$, the subsumption algorithm from~\cite{KS21,optimal_disc_sync} can be modified to detect termination. In our setting, with exactly $n$ agents in an $n$-node graph, we modify the dispersion algorithm to incorporate this termination detection. The key idea is that once multiple clusters have completed dispersion, it is guaranteed that, even if these clusters do not directly meet, a DFS tree emerging from any cluster can eventually reach another cluster by traversing an additional outgoing edge (i.e., an edge not internal to the cluster DFS itself). Since clusters disperse at different rates, a dispersed cluster may not immediately find a new agent, although a larger cluster can immediately find one soon after completing its own dispersion.

To begin, each cluster disperses locally according to the methodology of~\cite{optimal_disc_sync}. Agents progressively cover unvisited nodes through the expansion of a DFS tree until the \emph{covered configuration} is reached. When two DFS trees originating from different clusters encounter one another, a \emph{subsumption} operation takes place, in which one DFS tree is absorbed into the other, eventually forming a single \emph{covered} DFS tree.

Consider a DFS $i$ of size $k_i$, where $k_i$ denotes the number of settled agents in the DFS plus the number of seeker agents currently located at $head(i)$. Let $r_t$ be the last agent to settle in DFS $i$, placed at node $v$ whose parent is node $u$. After $r_t$ settles, the seeker team probes from $v$ to determine whether there exists an external edge from DFS $i$ leading either to an uncovered empty node or to a node covered by agents from another DFS. At this stage, two cases may arise:

Case 1. If such an external edge exists at $v$, $r_t$ returns to $u$ and informs the agent covering $u$ to extend its oscillation pattern: either by adding $v$ to its oscillation cycle (if $u$ is already covered by an oscillating agent) or by initiating oscillation between $u$ and $v$ (if $u$ is covered by a non-oscillating agent). This guarantees that $v$ remains covered. Subsequently, $r_t$ becomes a free agent and moves to the newly discovered node and becomes $head(i)$, setting $v$ as its $parent$ (the similar $child$ information is updated at the agent covering $v$ simultaneously). If the new node is empty, $r_t$ waits there until it is potentially reached by another DFS; otherwise, if $r_t$ encounters an agent from another DFS, the subsumption process begins immediately.

Case 2. If no external edge is found from $v$ to an empty node or to a node occupied by another DFS (meaning $v$ is a leaf node or all its edges lead to the DFS $i$ itself), the seeker team collects $r_t$ and continues backtracking until such a node has been discovered. During backtracking, the node $v$ remains covered through the modified oscillation sequence of its parent $u$. Let's assume that DFS $i$ has $k_i$ agents. Since the seeker team has size at least $\lceil \tfrac{k_i}{3} \rceil + 1$, it can probe its current node in $O(1)$ rounds to find a suitable external node while backtracking. Once an external node is identified, $r_t$ settles there in the same manner as described in Case 1. In this situation, the node (say, covered by an agent $r_w$) from which the external edge was detected is designated as the new $parent$, and the discovered node becomes the new $head(i)$. Accordingly, the pointers of $r_t$ and $r_w$ are updated. 

After the new $head$ has been assigned, the rest of the seeker agents now continue to complete the dispersion process following ~\cite{optimal_disc_sync}. Now, importantly, if during its backtracking process the seeker team reaches the root of DFS $i$ and still finds no external outgoing edge, this indicates that all $n$ nodes have been visited, all DFSs have been merged to a single DFS, and the dispersion can now complete. As and when this termination condition is detected, the seeker team settles at the remaining empty nodes as the root initiates a complete graph traversal, simultaneously informing all the agents that dispersion has been achieved and is complete. From the discussion above, we record the following lemma.

\begin{lemma}
\label{lem:guaranteed-meeting-overlap}
If a DFS $i$ performs one extra exploration to an uncovered node $x$, then (i) some other DFS $j$ must eventually reach $x$, and (ii) the waiting time of the new $head(i)$ at $x$ is absorbed into the dispersion time of DFS $j$.
\end{lemma}

\begin{proof}
If $x$ is already occupied when $head(i)$ arrives, the lemma follows immediately. Suppose instead that $x$ is empty and uncovered. Since there are exactly $n$ agents on $n$ nodes, each settled agent occupies a distinct node, and no DFS can contain more agents than the number of nodes it spans. Hence, as dispersion proceeds, every uncovered node must eventually be visited; in particular, some other DFS will expand to $x$, establishing~(i). For~(ii), observe that DFS $j$, which eventually reaches $x$ continues its dispersion concurrently while $head(i)$ is waiting. The time spent by $head(i)$ waiting at $x$ is therefore bounded by the remaining dispersion time of the cluster belonging to DFS $j$.  

\end{proof}

\subsubsection*{Dispersed Configuration:}

The fully dispersed case, where the system begins with exactly one agent per node, presents an additional challenge. In a non-dispersed setting, the presence of at least two agents within a DFS ensures progress: one agent continues exploration while the other maintains coverage of the vacated node, allowing all clusters to merge into a single DFS by successively acquiring the isolated agents, as discussed earlier. By contrast, in a dispersed configuration, every agent is isolated, and adjacent singleton agents have no immediate awareness of each other. This lack of coordination may lead to expansion conflicts—two neighbouring agents may simultaneously move toward one another, vacating their nodes without ever meeting, thereby stalling progress.  

To overcome this, we employ a lightweight symmetry-breaking procedure. Each singleton agent, at the start of the algorithm, sequentially scans the bits of its unique identifier from least to most significant bit. For each bit, if it is $1$, the agent moves through a fixed adjacent port (chosen arbitrarily but kept constant throughout the process) and then returns; if the bit is $0$, it remains stationary for two rounds. Agents terminate the procedure either upon joining a DFS or after completing the scan of all bits. Similarly, if an agent executing this procedure encounters an existing DFS tree, it immediately halts and integrates into the visiting DFS. Once an agent has fully scanned all bits of its identifier, it remains stationary unless later acquired by a DFS. Agents already belonging to a two-node DFS ignore further symmetry-breaking attempts by other neighbours (the visiting agent is instead absorbed into the existing tree). The procedure requires no knowledge of the maximum identifier length $\lambda$, as each agent terminates locally once its identifier bits are exhausted or it joins a DFS.

Although a meeting is not guaranteed for every pair of adjacent agents—as illustrated by the following example—some meeting is nevertheless inevitable. Suppose agents $a$ and $b$ are adjacent, with $a.ID = 1000$ and $b.ID = 1001$. During the bit-scanning process, each agent selects a fixed port at the start through which it moves whenever it encounters a $1$ bit in its identifier. Now, assume that $a$ selects the port leading to $b$, but $b$ selects a different port that leads elsewhere. At the least significant bit (LSB), $a$ observes a $0$ and remains stationary, while $b$ sees a $1$ and moves through its chosen port—away from $a$—and then returns. In the next two bit positions, both $a$ and $b$ see $0$s and remain stationary. At the final bit, both see a $1$ and move simultaneously—but along their independently fixed ports, again avoiding each other. Throughout the entire scan, $a$ and $b$ never co-occupy a node, and thus never detect one another, despite being adjacent and having differing identifiers. This scenario demonstrates that local symmetry-breaking may fail between specific neighbouring pairs, especially when movement directions are not mutually aligned.

Nonetheless, it is still possible to guarantee that \emph{some} meeting will inevitably occur. Specifically, we show that the agent with the maximum identifier, denoted $r_{\max}$, is guaranteed to meet one of its neighbours during the bit-scanning process. Since $r_{\max}$ has a uniquely larger identifier, there must exist a bit position where its bit is $1$ while the corresponding bit of the neighbouring agent reached through its fixed port is $0$. When this bit is scanned, $r_{\max}$ moves to that neighbour’s node while the neighbour remains stationary, resulting in a meeting. This ensures that a two-node DFS is initiated, from which the full DFS can grow as described earlier.

\begin{lemma}[Meeting Guarantee via Maximum ID]
Let $r_{\max}$ be the agent with the maximum identifier among all agents in a fully dispersed configuration. Then, during the bit-scanning symmetry-breaking process, $r_{\max}$ is guaranteed to meet the neighbour connected through its arbitrarily chosen movement port before completing its identifier scan, regardless of the neighbour's own movement behaviour.
\end{lemma}

\begin{proof}
Let $r_{\max}$ be the agent with the maximum identifier among all agents in the system. Let $b$ denote the agent located at the node connected to $r_{\max}$ through its arbitrarily chosen fixed port for the symmetry-breaking process. Since all identifiers are distinct and $r_{\max}$ holds the maximum, when comparing the binary strings of $r_{\max}$ and $b$ from the most significant bit (MSB) to the least, there must exist a first bit position $i^\star$ where $\text{bit}_{i^\star}(r_{\max}) = 1$ and $\text{bit}_{i^\star}(b) = 0$. If no such position existed, then $b$ would match or exceed $r_{\max}$ in all bit positions, contradicting the maximality of $r_{\max}$.

Now consider the behavior of both agents during the symmetry-breaking process. In the rounds corresponding to bit position $i^\star$, $r_{\max}$ will move through its fixed port toward $b$'s node and return, since its bit is $1$. Meanwhile, $b$, seeing a $0$ in the same position, remains stationary for both rounds. As a result, $r_{\max}$ will visit $b$'s node while $b$ is idle there, ensuring that a meeting occurs at that step. Therefore, a meeting is guaranteed before $r_{\max}$ completes its full identifier scan.
\end{proof}

\begin{lemma}\label{lem:stage1}
    The dispersion of the $n$ agents, the construction of the single DFS tree, and the detection of termination together complete in $O(n)$ rounds. 
\end{lemma}

\begin{proof}
First of all, since each singleton cluster initiates a symmetry-breaking mechanism at the start, a maximum time of $O(\log\lambda)=O(\log n)$ (worst case scenario, for a dispersed configuration) is required at first before the actual dispersion begins (this symmetry-breaking mechanism is independent and does not affect a dispersion process happening elsewhere). Then, the algorithm consists of four phases: dispersion of individual clusters, subsumption of DFS trees, backtracking to select an external edge and a final traversal. By~\cite{optimal_disc_sync}, the dispersion of a cluster of size $k_i$ completes in $O(k_i)$ rounds, and any backtracking to search for external nodes can also be performed in $O(k_i)$ rounds due to parallel probing. When two DFS trees of sizes $k_i$ and $k_j$ meet, the subsumption process takes $O(\max\{k_i,k_j\})$ rounds~\cite{KS21}. Once all clusters have been subsumed into a single DFS tree, the final traversal from the root to visit every node requires $O(n)$ rounds, as such traversals can be carried out without incurring significant memory overhead, by equipping each child agent with a $sibling$ pointer, as employed in~\cite{butterfly_spaa_prabhat,agent_bfs,optimal_disc_sync}. Now, in the worst case, subsumptions occur sequentially, with a single DFS tree absorbing others one by one. For a cluster of size $k_i$, the total cost of dispersion, subsumption and backtracking is $O(k_i)$, and summing over all clusters yields $\sum_i O(k_i) = O(n)$. Hence, the overall time complexity of Stage~1 is $O(\log n)+O(n)=O(n)$ rounds.  
\end{proof}

\begin{remark}[New Oscillation]
    Since the algorithm adds one extra oscillation for the last settled agent before exploring an external edge in the DFS, each agent waits $8$ rounds (instead of $6$) at a node to verify whether it is covered.  
\end{remark}
\begin{remark}[Leader Election]
    The root agent at the end of the Stage $1$ can function as an elected leader. 
\end{remark}

\subsubsection{Stage 2: Gathering via the Spanning Tree.}  
To gather all the agents at the root, the process proceeds as follows. The gathering begins at the leaf agents of the spanning tree, which, having no children, move directly to their parent and accumulate there. At any intermediate stage, consider an agent $a$ that serves as an internal node of the tree. Agent $a$ remains stationary until it has received all the agents collected by its children. Once this condition is satisfied, $a$ moves to its parent, carrying along all the gathered agents. By repeating this process up the tree, all agents are eventually collected at the root node. When the root has received all the $n-1$ other agents, the gathering phase is complete, and we can proceed to the final stage of the algorithm, as we note the following lemma from our discussion.   

\begin{lemma}\label{lem:stage2}
    Gathering via the Spanning Tree takes $O(n)$ rounds.
\end{lemma}

\subsubsection{Stage 3: Minimal Dominating Set Computation.}  
With all $n$ agents assembled at the root, the agents collaboratively execute the rooted-case minimal dominating set (mDS) algorithm described in Section~\ref{rooted}, thereby computing the final minimal dominating set of the graph.  We therefore, obtain the following theorem from lemmas~\ref{lem:guaranteed-meeting-overlap},\ref{lem:stage1},\ref{lem:stage2}.

\begin{theorem}\label{thm:arbitrary}
    Let $G$ be an arbitrary connected simple anonymous graph with $n$ nodes. Suppose $n$ autonomous mobile agents are initially placed arbitrarily across the graph. Then, the agents can identify a minimal dominating set (coloured \textcolor{red}{red}) in $O(n)$ rounds using only $O(\log n)$ bits of memory per agent.
\end{theorem}

\section{Conclusion and Future Directions}\label{conclusion}
We proposed two linear-time algorithms for computing a minimal dominating set using mobile agents in both rooted and arbitrary configurations, achieving $O(n)$ round complexity in each case. The framework also facilitates linear-time solutions to related problems such as leader election, gathering, and spanning tree construction, all without requiring any global knowledge.

Future work includes reducing the number of participating agents — for instance, initiating the computation with only those located on a minimal dominating set and enabling them to oscillate to simulate communication. Another challenge is to develop mechanisms that operate under asynchronous settings where agents may not meet systematically, causing delays or incorrect outcomes. In addition, another important direction is to design fault-tolerant variants that compute a minimal dominating set despite crash or Byzantine failures.

\subsection*{Acknowledgements}
The authors thank \textbf{Dr.~Manish Kumar}, \textit{Postdoctoral Researcher} at \textit{IIT Madras}, for his useful comments on the \emph{leader election} in the dispersed-configuration case in Section~\ref{arbitrary}, which led to a more refined analysis.

\bibliographystyle{unsrt}
\bibliography{references}
\end{document}